\documentclass[11pt,a4paper]{article}
\usepackage[latin1]{inputenc}
\usepackage{amsmath}
\usepackage{amsfonts}
\usepackage{amssymb}
\usepackage{graphicx}
\usepackage{amsmath,amssymb,amsthm,color,graphicx, fancybox}
\usepackage{tcolorbox}
\usepackage{bclogo}
\usepackage[latin1]{inputenc}
\usepackage{mdframed}
\usepackage{amsmath}
\usepackage{amsfonts}
\usepackage{amssymb}
\usepackage{wasysym}
\usepackage{amsthm}
\usepackage[english]{babel}

\usepackage{subfigure}% subcaptions for subfigures
\usepackage{subfigmat}% matrices of similar subfigures, aka small mulitples
\usepackage{float}

\usepackage{subfigure}% subcaptions for subfigures
\usepackage{subfigmat}% matrices of similar subfigures, aka small mulitples
\usepackage{float}

\newcommand{\N}{\mathbb{N}}

\newcommand{\I}{{\cal I}}

\newcommand{\R}{\mathbb R}

\newtheorem{conj}{Conjecture}
\newtheorem{theorem}{Theorem}[section]
\newtheorem{lemma}{Lemma}[section]
\newtheorem{prop}{Proposition}[section]
\newtheorem{cor}{Corollary}[section]

\newtheorem{remark}{Remark}[section]
\newtheorem{defi}{Definition}[section]
\newtheorem{acknowledgment*}{Acknowledgment}
\newtheorem{assumption}{Assumption}[section]

\newcommand{\be}{\begin{equation}}
\newcommand{\ee}{\end{equation}}

\title{\Large \textbf{ Happy family of stable marriages} }
\author{\textsc{Gershon Wolansky }\footnote{Department of Mathematics, Technion, Israel Inst. of Technology}}
\begin{document}

\maketitle
\tableofcontents
%\begin{document}
\abstract{
In this chapter we study some aspects of the problem of stable marriage. There are two distinguished marriage plans:  the fully transferable case, where money can be transferred between the participants, and the fully non transferable case where each participant has its own rigid preference list regarding the other gender. We continue to discuss intermediate  partial transferable cases.  Partial transferable plans can be approached as either special cases of cooperative games using the notion of a core,  or as a generalization of the cyclical monotonicity property of the fully transferable case (fake promises). We shall introduced these two approaches, and prove the existence of stable marriage for the fully transferable and non-transferable plans. 
}

\section{Introduction}
	Consider two sets $\I_m,\I_w$  of  $N$ elements each. We may think about $\I_m$ as a set of men and $\I_w$ as a set of women. We denote a man in $\I_m$ by $i$ and a woman in $\I_w$ by $i^{'}$.
	
	A {\em marriage  plan} (MP) is a bijection which assign to each man in $\I_m$ a unique woman in $\I_w$ (and v.v). A matching of a man $i\in\I_m$ to a woman $j^{'}\in\I_w$ is denoted by $ij^{'}$.  The set of all such matchings is isomorphic  to the set of permutations on $\{1, \ldots N\}$. Evidently, we can arrange the order according to a given marriage plan  and represent this plan       as     $ \{ii^{'}\}; i=1\ldots N$.

	The MP $\{ii^{'}\}$ is called {\em stable} if and only if there are no {\em blocking pairs}. A blocking pair is composed of a man $i$ and a woman $j^{'}\not=i^{'}$ such that {\em both} $i$ prefers $j^{'}$ over his assigned woman $i^{'}$ and $j^{'}$ prefers $i$ over her assigned man $j$.
	
	In order to complete this definition we have to establish a criterion of preferences over the possible matchings in $\I_m\times \I_w$. 
	%@Suppose that every man has an order of preference on the women, and every woman has an order of preference on the men. 
	\par\noindent
	Let us consider  two extreme cases. The first is the {\em fully transferable} (FT) case \cite{shap1, SS, Bec1,Bec2}. Here we assume a {\em utility value} $\theta_{ij^{'}}$ for a potential  matching $ij^{'}$.  If $ij^{'}$ are matched, 
	they can split this reward $\theta_{ij^{'}}$ between themselves as they wish. 
	%The man $i$  may cut $u=\alpha\theta_{ij^{'}}$ out of this reward, while the woman $j^{'}$ may cut $\beta\theta_{ij^{'}}$, where $\alpha+\beta\leq 1$. 

	The second case is fully non-transferable (FNT) \cite{GSH, knu, Dub}. This involves  no utility value (and no money reward). Each participant (man or woman) list the set of participants of the other gender according to a preference list: For each man $i\in\I_m$ there exist an order relation $\succ_i$ on $\I_m$, such that $j^{'}\succ_i k^{'}$ means that the man $i$ will prefer the woman $j^{'}$ over the woman $k^{'}$. Likewise, each woman $i^{'}\in\I_w$ have its own order relation $\succ_{i^{'}}$ over $\I_m$. 
	
	These two notions seems very different, and indeed they are, not only because the first one seems to  defines the preference  in materialistic terms and the second hints on "true love". In fact, we can quantify 
	the non-transferable case as well: There may be a reward $\theta^m_{i j^{'}}$  for a man $i$ marrying a woman $j{'}$, such that $j^{'}\succ_ik^{'}$ iff $\theta^m_{ij^{'}}>\theta^m_{ik^{'}}$. Likewise, 
	$\theta^w_{ij^{'}}$ quantifies the reward the the woman $j^{'}$ obtains while marrying the man $i$.
	
	Given a matching $\{ii^{'}\}$, a blocking pair in the FNT case is a pair $ij^{'}$, $j^{'}\not =i^{'}$ such that
	the man $i$ prefers the woman $j^{'}$ over his matched woman
	%\end{document}
	$i^{'}$ (i.e $j^{'}\succ_ii^{'}$, or $\theta^m_{ij^{'}}>\theta^m_{ii^{'}}$ ) {\em and} the woman $j^{'}$ prefers $i$ over her matched man $j$ ($j\succ_{i^{'}}i$, or $\theta^w_{ij^{'}}>\theta^w_{jj^{'}}$ ). Thus, a blocking pair
	$ij^{'}$ is defined by 
	\be\label{bp1} \min\{ \theta^m_{ij^{'}}-\theta^m_{ii^{'}}, \ \theta^w_{ij^{'}}-\theta^w_{jj^{'}}\}>0 \ . \ee
	\begin{defi}\label{firstd}
		The matching $\{ii^{;}\}$ is stable if and only if 
		$$ \min\{ \theta^m_{ij^{'}}-\theta^m_{ii^{'}}, \ \theta^w_{ij^{'}}-\theta^w_{jj^{'}}\}\leq 0$$
		for any $i,j\in\I_m$ and $i^{'}, j^{'}\in\I_w$.  
	\end{defi}
	Let 
	\be\label{sum}\theta_{ij^{'}}:=\theta^m_{ij^{'}}+\theta^w_{ij^{'}} \ . \ee
	Definition \ref{firstd} implies that the condition 
	\be\label{weak}\theta_{ii^{'}}+\theta_{jj^{'}}\geq \theta_{ij^{'}} + \theta_{ji^{'}} \  \ee
	is {\em necessary} for all $i,j$ for the stability of $\{ii^{'}\}$ in the FNT case.  
	
	Let us consider now the fully transferable (FT) case. Here a married pair $ii^{'}$ can share the rewards $\theta_{ii^{'}}$ for their marriage. 
	Suppose the man $i$  cuts $u_i$  and the woman $i^{'}$ cuts $v_{i^{'}}$ form their mutual reward $\theta_{ii^{'}}$. Evidently, $u_i+v_{i^{'}}= \theta_{ii^{'}}$. If 
	\be\label{bp2} u_i+v_{j^{'}} < \theta_{ij^{'}} \ \ee
	for some $j^{'}\not= i^{'}$	then $ij^{'}$ is a blocking pair, since both $i$ and $j^{'}$ can increase their cuts to match the mutual reward $\theta_{ij^{'}}$. 
	Hence 
	$$ \theta_{ij^{'}}+\theta_{ji^{'}} > u_i+v_{j^{'}}+u_j+v_{i^{'}} = \theta_{ii^{'}}+\theta_{jj^{'}}$$
	so (\ref{weak}) is a necessary condition for 
	the stability in the FT case as well.
	
	Evidently, condition (\ref{weak}) is {\em not} a sufficient one, unless $N=2$ in the FT case.

	\begin{tcolorbox}
		
		A simple example ($N=2$):
		\begin{center}
			\begin{tabular}{ c c c }
				$\theta^m$& $w_1$& $w_2$\\
				$m_1$& 1 & 0 \\
				$m_2$ & 0 & 1
			\end{tabular}
			\ \ \ \ ; \ \ \
			\begin{tabular}{ c c c }
				$\theta^w$& $w_1$ & $w_2$\\
				$m_1$ & 1 & 5 \\
				$m_2$& 0 & 1
			\end{tabular}
		\end{center}
		The matching $\{ 11^{'},  22^{'}\}$ is FNT stable. Indeed  $\theta^m_{11^{'}}=1> \theta^m_{12^{'}}=0$
		while $\theta^m_{22^{'}}=1>\theta^m_{21^{'}}=0$, so both men are happy, and this is enough for FNT stability, since   that  neither $\{ 12^{'}\}$ nor $\{21^{'}\}$ is a blocking pair. On the other hand, if the married pairs share their rewards
		$\theta_{ij^{'}}=\theta^m_{ij^{'}}+\theta^w_{ij^{'}}$ we get
		\begin{center}
			\begin{tabular}{ c c c }
				$\theta$& $w_1$& $w_2$\\
				$m_1$& 2 & 5 \\
				$m_2$ & 0 & 2
			\end{tabular}
		\end{center}	
		so
		$$ \theta_{11^{'}}+\theta_{22^{'}}=4< 5=\theta_{12^{'}}+\theta_{21^{'}}\ , $$
		thus $\{21^{'}, 12^{'}\}$ is the  stable marriage in the  FT case.
	\end{tcolorbox}
	However, we may extend the necessary condition  (\ref{weak}) in the FT case  as follows:
	
	Consider the couples 
	$i_1i_1^{'}, \ldots i_ki_k^{'}$, $k\geq 2$ . The sum of the rewards for these couples is
	$\sum_{l=1}^k\theta_{i_li^{'}_l}$.
	Suppose they perform a  "chain deal" such that man $i_l$  marries woman $i^{'}_{l+1}$ for $1\leq l\leq k-1$, and the last  man $i_k$ marries the first woman $i^{'}_1$. The  net reward for the new matching is $\sum_{l=1}^{k-1}\theta_{i_li^{'}_{l+1}}+ \theta_{i_ki^{'}_1}$.
	
	This  leads to a definition of a {\em blocking chain}:
	\begin{defi}\label{chain}
		A  chain $i_1i^{'}_{1}, \ldots i_{k} i^{'}_{k}$ of married couples forms a blocking chain iff 
		\be\label{chain*} \sum_{l=1}^k(\theta_{i_li^{'}_{l+1} }-\theta_{i_li^{'}_l} )>0\  \ee
		where $i^{'}_{k+1} :=i^{'}_{1}$.  If there are no blocking chains then the matching $\{ii^{'}\}$ is called {\em cyclically monotone} \cite{Roc}.
	\end{defi}
	The notion of a blocking chain extends the condition (\ref{bp2}) from $k=2$ to $k\geq 2$. It turns that it is  also necessary condition for the stability in the fully  transferable case:
	\begin{prop}\label{t01}
		If a marriage $\{ii^{'}\}$ is a stable one for the FT case  then it is cyclically monotone. 
	\end{prop}

	\begin{proof}
		Let $\{ii^{'}\}$ be a matching, such that $u_i$ is the cut of man $i$ marrying $i^{'}$ and $v_{i^{'}}$ the cut of the woman $i^{'}$ marrying $i$. Suppose by negation that  $i_1i^{'}_1\ldots i_ki^{'}_k$ is a blocking  chain. 
		Since $u_i+v_{i^{'}}\leq \theta_{ii^{'}}$ we obtain
		$$\sum_{l=1}^k\theta_{i_li^{'}_{l+1}}> \sum_{l=1}^k\theta_{i_li^{'}_{l}} \geq \sum_{l=1}^k(u_{i_l} +  v_{i^{'}_l})=\sum_{l=1}^k(u_{i_l} +  v_{i^{'}_{l+1} }) $$
		so, in particular, there exists a pair $i_li^{'}_{l+1}$ for which $\theta_{i_li^{'}_{l+1}} > u_{i_l} + v_{i^{'}_{l+1}}$. Hence $i_li^{'}_{l+1}$ is a blocking pair via (\ref{bp2}).
	\end{proof}
	We shall see later on that cyclical monotonicity is, actually, an {\em equivalent definition} to stability in the FT case. 
	
	The notion of cyclical monotonicity implies an additional level of cooperation for the marriage game. Not only the married pair share their utility  between themselves via (\ref{sum}), but also different couples are ready to share their reward via a chain deal according to Definition \ref{chain}. If the total reward after the chain exchange exceeds their reward prior to this deal, the lucky ones are ready to share their reward with the unlucky and  compensate their loss
	
	What about the FNT case? Of course there is no point talking about a "chain deal" in that case. However, we may define a "FNT blocking chain"  $i_1i^{'}_1\ldots i_{k}i^{'}_k$ by 
	\be\label{bpNT}\max_{1\leq l\leq k} \min\{ \theta^m_{i_li^{'}_{l+1}}-\theta^m_{i_li^{'}_l},  \theta^w_{i_li^{'}_{l+1}}-\theta^w_{i_li^{'}_l}\}>0\ee
	where, again, $i^{'}_{k+1}\equiv i^{'}_1$.  Definition \ref{firstd} is analogs to the statement that that there are no blocking chains of this form. Thus, a marriage $\{ii^{'}\}$ is stable in the FNT case if and only if
	\be\label{chain**} \max_{1\leq l\leq k} \min\{ \theta^m_{i_li^{'}_{l+1}}-\theta^m_{i_li^{'}_l},  \theta^w_{i_li^{'}_{l+1}}-\theta^w_{i_li^{'}_l}\}\leq 0\ee
	for any chain deal $i_1i^{'}_1\ldots i_{k}i^{'}_k$.
	
	At the first sight definition (\ref{chain**}) seems redundant, since it provides no further information. However, 
we can observe the  analogy between (\ref{chain*}) and (\ref{chain**}).  In fact, (\ref{chain**}) and  (\ref{chain*}) are obtained from  each other  by the exchanges
	\be\label{dual} \theta_{i_li^{'}_{l+1}}-\theta_{i_li^{'}_l} \Longleftrightarrow \min\{ \theta^m_{i_li^{'}_{l+1}}-\theta^m_{i_li^{'}_l},  \theta^w_{i_li^{'}_{l+1}}-\theta^w_{i_li^{'}_l}\}\ \ \ 
	\text{and} \ \ \  \sum_1^k \Longleftrightarrow \max_{1\leq i\leq k}\ee
	In section \ref{fake} below we will take advantage on this  representation. 
	\section{Partial sharing}
	Here we present two possible definitions of intermediate marriage game which interpolate between the fully transferable and the    non transferable case. 
	The first is based on the notion of core of  a cooperative game, and the second is based on cyclic monotonicity. 
	\subsection{Stable marriage as  a cooperative game}\label{gencase}
	This part follows some of the ideas in Galichon et.all and references therein\footnote{which was turned to my attention by R. McCann} \cite{GASD}.  See also \cite{ECA}.

	Assume that  we can guarantee a  cut  $u_i$  for each married man $i$, and a cut $v_{j^{'}}$ for  each married woman $j^{'}$. In order to define a stable marriage we have to impose some conditions which will guarantee that no man or woman can increase his or her cut by marrying a different partner. For this let us define, for each pair $ij^{'}$, a {\em pairwise bargaining set} ${\cal F}(ij^{'})\subset \R^2$ which contains all possible cuts $(u_i,v_{j^{'}})$ for a matching of man $i$ with woman $j^{'}$.
	\begin{assumption}\label{genass}  \ .
		\begin{description}
			\item{i)}
			For each $i\in \I_m$ and $j^{'}\in \I_w$, ${\cal F}(ij^{'})$ are closed sets in $\R^2$, equal to the closure of their interior. Let  ${\cal F}_0(ij^{'})$ the interior of ${\cal F}(ij^{'})$.
			\item{ii)} ${\cal F}(ij^{'})$ is monotone in the following sense:
			If $(u ,v)\in {\cal F}(ij^{'})$ then $(u^{'}, v^{'})\in {\cal F}(ij^{'})$ whenever $u^{'}\leq u$ and $v^{'}\leq v$.
			\item{iii)} There exist $C_1, C_2\in\R$ such that $$ \left\{(u,v);  \max(u,v)\leq C_2   \right\}\subset {\cal F}(ij^{'})\subset \left\{(u,v); u+v\leq C_1  \right\} \ $$
			for any $i\in\I_m, j\in \I_w$.
		\end{description}
	\end{assumption}
	\par
	The meaning of the feasibility set is as follows:
	\begin{tcolorbox}
		% \begin{tcolorbox}
		Any married couple  $ij^{'}\in\I_m\times \I_w$ can guarantee the cut $u$ for $i$ and $v$ for $j^{'}$, provided $(u,v)\in {\cal F}(ij^{'})$.
	\end{tcolorbox}

	\begin{defi}\label{deffea} \ .
		The {\em feasibility set} $V({\cal F})\subset \R^{2N}$ is composed of all vectors $(u_1, \ldots u_N, v_1, \ldots v_N)$ which satisfies
		$$ (u_i,v_{j^{'}})\in \R^2-{\cal F}_0(ij^{'})$$
		for any $ij^{'}\in\I_m\times \I_w$. 
		
		The marriage plan $\{ii^{'}\}$ is stable if and only if there exists $(u_1, \ldots v_N)\in V({\cal F})$  such that $(u_i, v_{i^{'}})\in {\cal F}(ii^{'})$ for any $i\in\{1, \ldots N\}$.
	\end{defi}
	
	The FNT case is contained in definition \ref{deffea}, where 
	\be\label{Fblockp} {\cal F}(ij^{'}):= \{u\leq \theta^m_{ij^{'}} ; \ \ v\leq \theta^w_{ij^{'}} \}\ . \ee
	Indeed, if $\{ii^{'}\}$ is a stable marriage plan let $u_i=\theta^m_{ii^{'}}$ and $v_{i^{'}}=\theta^w_{ii^{'}}$. Then $(u_1, \ldots v_N)$ satisfies $(u_i,v_{i^{'}})\in{\cal F}$ for any $i\in\{1\ldots N\}$. Since there are no blocking pairs if follows that for any $j^{'}\not= i^{'}$, either $\theta^m_{ij^{'}}>\theta_{ii^{'}}=u_i$ or $\theta^w_{ij^{'}}>\theta^w_{jj^{'}}=v_{j^{'}}$, hence 
	$(u_i, v_{j^{'}})\in\R^2-{\cal F}_0(ij^{'})$ so $(u_1\ldots v_N)\in V({\cal F})$ (Fig. [1a]).

	The FT case  (Fig. [1b]) is obtained by
	\be\label{stfultrans} {\cal F}(ij^{'}):= \{(u,v); \ u+v\leq \theta_{ij^{'}} \} \ . \ee
	Indeed, if $\{ii^{'}\}$ is a stable marriage plan and $(u_1, \ldots v_N)$ are the corresponding cuts satisfying $u_i+v_{j^{'}}=\theta_{ij^{'}}$, then for each $j^{'}\not= i^{'}$ we obtain 
	$u_i+v_{j^{'}}\geq \theta_{ij^{'}}$ (otherwise $ij^{'}$ is a blocking pair). This implies that 
	$(u_i,v_{j^{'}})\in \R^2-{\cal F}_0(ij^{'})$.

\begin{figure}
	\begin{subfigmatrix}{4}% number of columns
			\subfigure[FNT]{\includegraphics[width=0.44\textwidth]{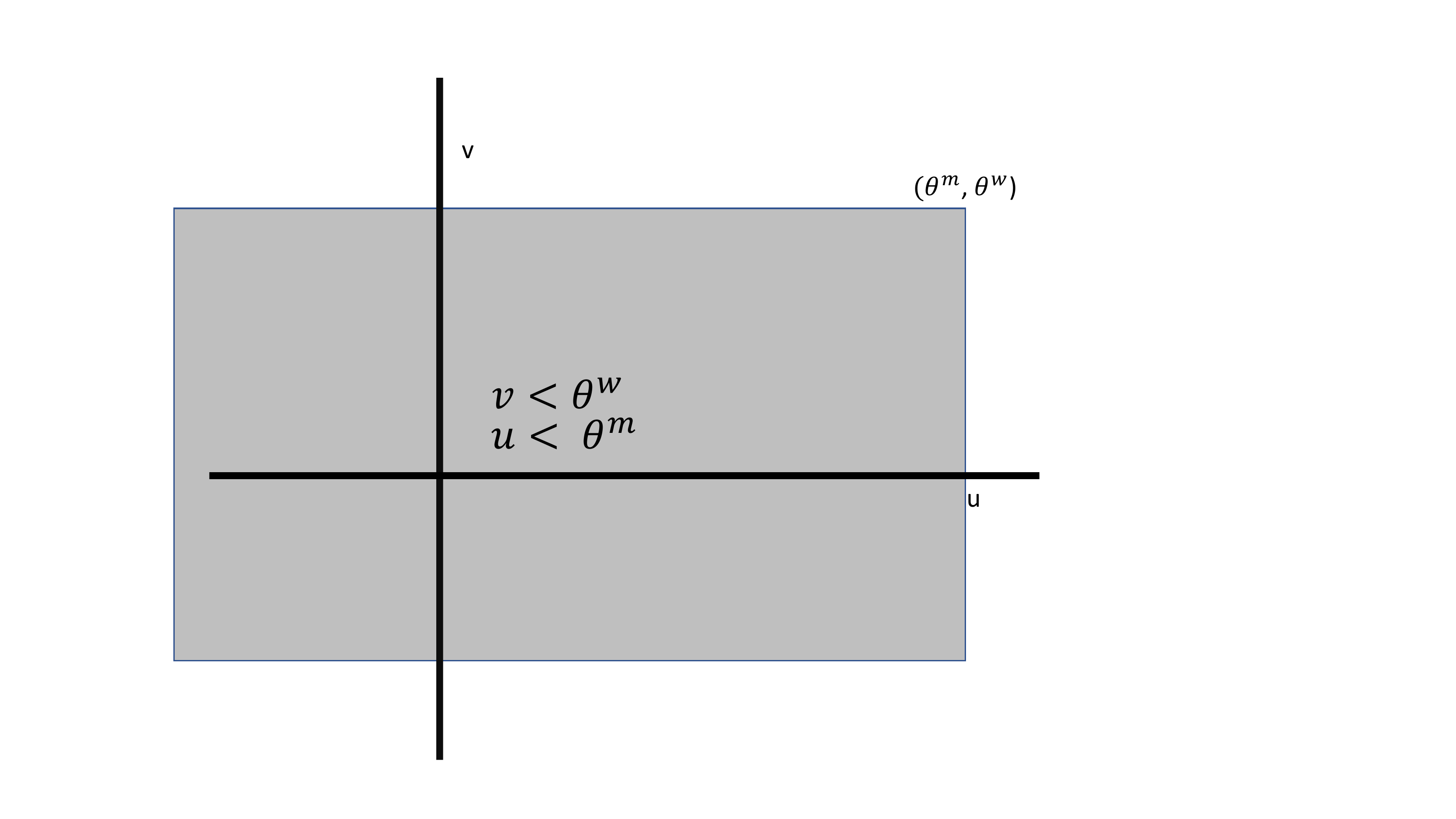}}
		\subfigure[FT]{\includegraphics[width=0.44\textwidth]{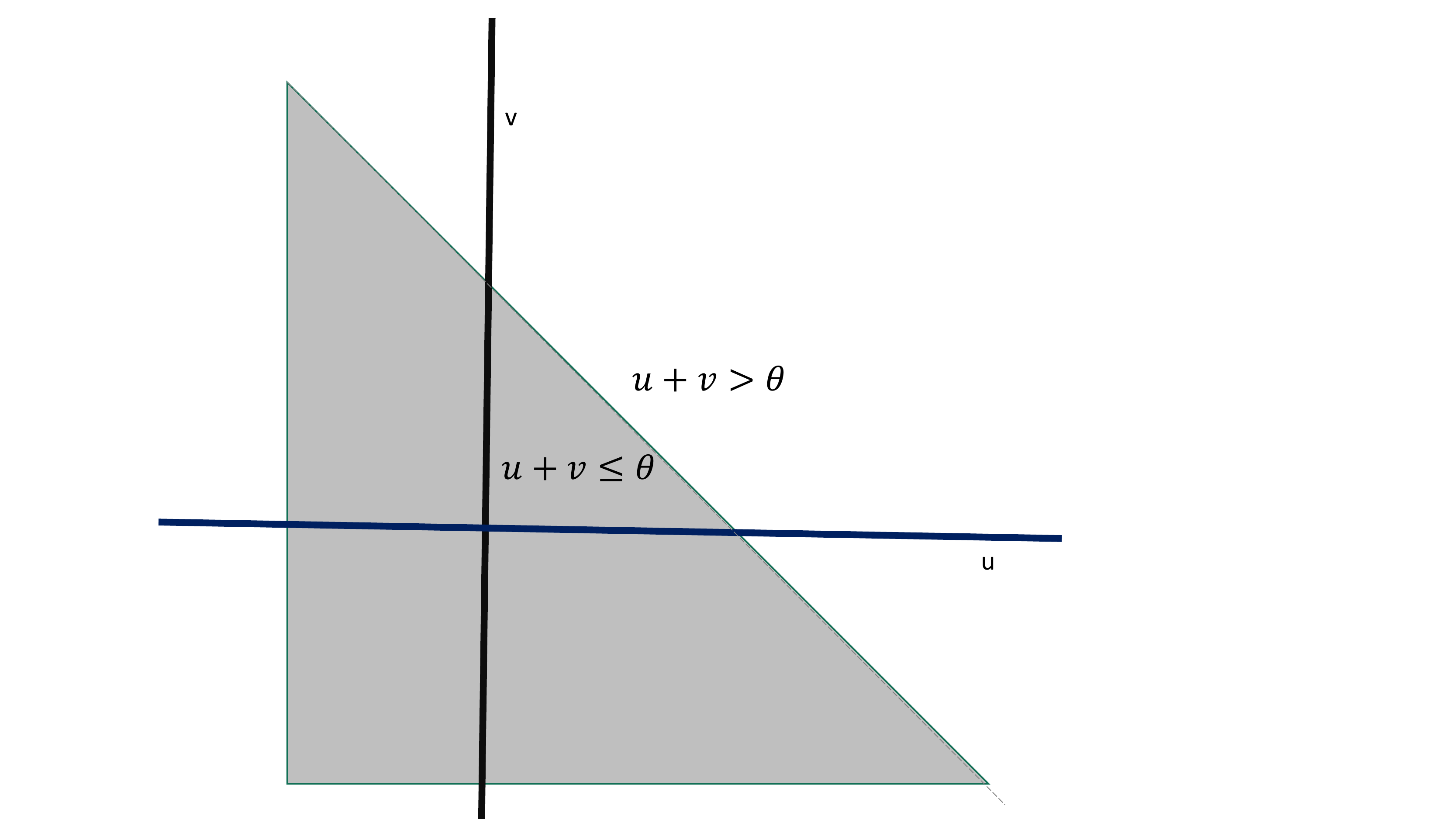}}
		%	\caption{a}
	
		%	\caption{b}
		\subfigure[Case 1]{\includegraphics[width=0.44\textwidth]{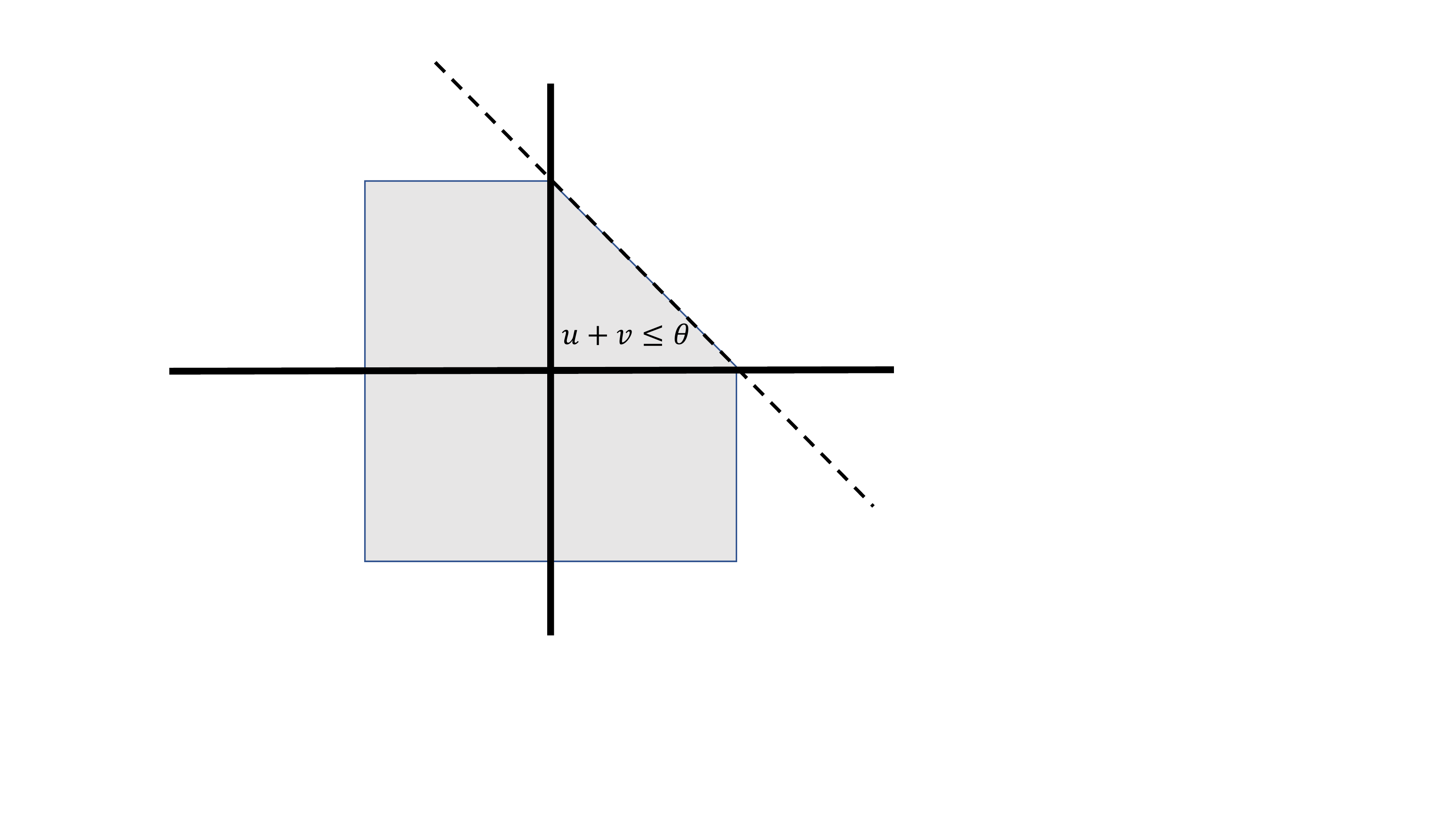}}
			\subfigure[Case 2]{\includegraphics[width=0.44\textwidth]{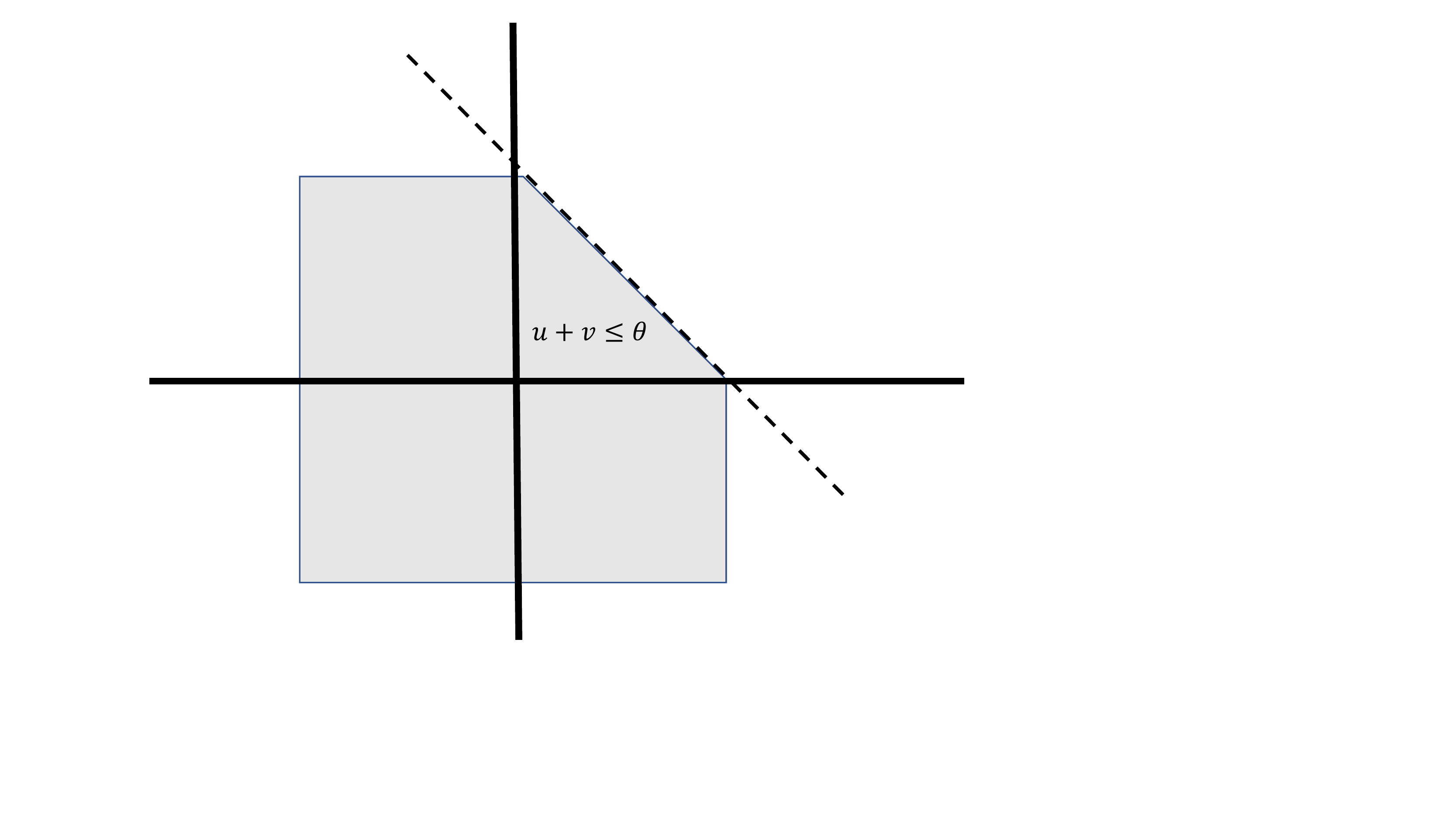}}
	%	\subfigure[]{\includegraphics[width=0.32\textwidth]{noconflict1.jpg}}
		\label{fig1}
	\end{subfigmatrix}
	%\caption{The 3 cases of conflict free chemotaxis for $\beta>0$: $\beta<\alpha/2$ (a),  $\beta>\alpha/2$, $\gamma=0$ (b) and  $\beta>\alpha/2$,   $\gamma>0$ (c). The vertical heavy line represents the critical mass $M_1=8\pi/\alpha$.}
		 \caption{pairwise bargaining sets}
\end{figure}
	%	\caption{The 3 cases of conflict free chemotaxis for $\beta>0$: $\beta<\alpha/2$ (a),  $\beta>\alpha/2$, $\gamma=0$ (b) and  $\beta>\alpha/2$,   $\gamma>0$ (c). The vertical heavy line represents the critical mass $M_1=8\pi/\alpha$.}

%	\end{figure}
	There are other, sensible models of  {\em partial transfers} which fit into the formalism of Definition \ref{deffea} and Theorem \ref{cornoemp}.  Let us consider several examples:
	\begin{enumerate}
		\item {\em Transferable marriages restricted to non-negative cuts} : In the transferable case the feasibility sets may contain negative cuts for the man $u$ or for the woman $v$ (even though not for both, if it is assumed $\theta_{ij^{'}}>0$). To avoid the undesired stable marriages were one of the partners get a negative cut we may replace the feasibility set (\ref{stfultrans}) by
		$$ {\cal F}(ij^{'}):=\{(u,v)\in\R^2; u+v\leq \theta_{ij^{'}}\ ,  \max(u,v)\leq \theta_{ij^{'}}\} \ , $$
		see Fig. [1c]. 
		It can be easily verified that if $(u_1, \ldots v_N)\in V({\cal F})$ contains negative components, then $([u_1]_+, \ldots [v_N]_+)$, obtained by replacing the negative components by $0$, is in $V({\cal F})$ as well. Thus, the core of this game contains vectors in $V({\cal F})$ of non-negative elements.
		\item In the transferable case (\ref{stfultrans}) we allowed both men and women to transfer money to their partner. Indeed, we assumed that the man's $i$ cut is $\theta^m_{ij^{'}}-w$ and the woman's $j$ cut is  $\theta^w_{ij^{'}}+w$, where $w\in\R$. Suppose we wish to allow only transfer between men to women, so we insists on $w\geq0$.
		In that case we choose (Fig. [1d])
		\be\label{calFFT} {\cal F}(ij^{'}):=\{ (u,v)\in\R^2; \ u+v\leq \theta_{ij^{'}}; \ \ u\leq \theta^m_{ij^{'}} \}  \  .  \ee
		\item
		Let us assume that the transfer $w$ from man $i$ to woman $j^{'}$ is taxed, and the tax depends on $i,j^{'}$. Thus, if man $i$ transfers $w>0$ to a woman $j^{'}$  he reduces his cut by $w$, but the woman cut is increased by an amount $\beta_{i,j}w$, were $\beta_{i,j}\in [0,1]$. Here  $1-\beta_{i,j}$ is the tax implied for this transfer. It follows that
		$$ u_i\leq \theta^m_{ij^{'}}-w \ \ \ ; \ \ \ v_{j^{'}}\leq \theta^w_{ij^{'}}+\beta_{i,j}w \ , \ \ w\geq 0 $$
		Hence
		$$ {\cal F}(ij^{'}):=\{(u,v)\in\R^2; \ u_i + \beta^{-1}_{i,j}v_{j^{'}}\leq  \theta^\beta_{ij^{'}} , \ \ u_i\leq \theta^m_{ij^{'}}\} \ ,  $$
		where $\theta^\beta_{ij^{'}}:= \theta^m_{ij^{'}} + \beta^{-1}_{i,j}\theta^w_{ij^{'}}$. The geometrical description of ${\cal F}$ us as in Fig. [1d], where the dashed line is tilted. 
	\end{enumerate}
	\subsection{Stability by fake promises} \label{fake} 

Suppose a man can make a promise to a married woman (which is not his wife), and v.v. The principle behind it is that each of them does not intend to honor his/her own promise, but, nevertheless, believes that the other parti will honor hers/his.
It is also based on both partial sharing inside a married pair, as well as some collaboration between the pairs.
\par
Define 
\be\label{introq}  \Delta^{(q)}(i,j^{'}):=\min\left\{ \begin{array}{c}
	q(\theta^m_{ij^{'}}-\theta^m_{ii^{'}})+\theta^w_{ij^{'}}-\theta^w_{jj^{'}} \\
	q(\theta^w_{ij^{'}}-\theta^w_{jj^{'}}+\theta^m_{ij^{'}}-\theta^m_{ii^{'}}
\end{array}\right\} \ ,
\ee
where $0\leq q\leq 1$. In particular
$$
\Delta^{(0)}(i,j^{'}):=\min\{ \theta^m_{ij^{'}}-\theta^m_{ii^{'}}, \theta^w_{ij^{'}}-\theta^w_{j^{'}j^{'}}\} $$ 
$$ \Delta^{(1)}(i,j^{'}):=
\theta^m_{ij^{'}}-\theta^m_{ii^{'}}+ 
\theta^w_{ij^{'}}-\theta^w_{ii^{'}}\equiv
\theta_{ij^{'}}-\theta_{ii^{'}}
. $$
The value of $q$ represents the level of {\em internal sharing} inside the couple. Thus, $q=0$ means there is no sharing whatsoever, and the condition  $\Delta^{(0)}(i,j^{'})>0$ for a blocking pairs implies that both $i$ and $j^{'}$ gains from the exchange, is displayed in  (\ref{bpNT}) . 

On the other hand, $\Delta^{(1)}(i, j^{'})+\Delta^{(1)}(j,i^{'})>0$, namely 
$$\theta_{ii^{'}}+ \theta_{jj^{'}}< \theta_{ij^{'}}+\theta_{ji^{'}}$$
is, as we argued,  a necessary condition for a blocking pair in FT case, where $\theta$ represents the sum of the rewards to of the pair via (\ref{sum}). 

We now consider an additional parameter $p\in[0,1]$ and define the real valued function on $\R$:
\be\label{pospdef} x \mapsto [x]_p:= [x]_+ - p[x]_- \  \ee
Note that $[x]_p=x$  for any $p$ if $x\geq 0$, while $[x]_1=x$ for any real $x$. The parameter $p$ represents the level of sharing {\em between the pairs}. 
\begin{defi}\label{defshapleyq} Let \ $0\leq p,q\leq 1$. The matching $\{ii^{'}\}$  is $(p,q)-$stable if for any $k\in\N$ and  $i_1, i_2, \ldots i_k\in \{1, \ldots N\}$
	$$ \sum_{l=1}^k\left[ \Delta^{(q)}(i_l, i^{'}_{l+1})\right]_p  \leq 0 \  \text{where} \ \ i_{k+1}=i_1 \  $$
	where $i_{k+1}:= i_1$.
	%Here $[a]_+:=a\vee 0$, $[a]_-:= -a\vee 0$.
\end{defi}
%	\begin{tcolorbox}
Note that  $p=0$  implies that $\Delta^{(q)}(i, j^{'})\leq 0$ for {\em any} $j^{'}\not= i^{'}$. If, in addition, $q=0$ then this is just the statement that there are no blocking pairs in the FNT case. 

On the other hand, $p=1$ implies
$$ \sum_{l=1}^k \Delta^{(q)}(i_l, i^{'}_{l+1})  \leq 0 \  \text{where} \ \  $$
which is reduced to (\ref{chain*}) if $q=1$ as well. 

\par
Let us interpret the meaning of $q,p$ in the context of utility exchange. 
A man $i\in\I_m$ can offer  some bribe $w$   to any  other women $j^{'}$ he might  be  interested in (except his own wife, so $j^{'}\not=i^{'}$). His cut for marrying $j^{'}$ is now
$\theta^m_{ij^{'}}-w$. The cut of the woman $j^{'}$ should have been $\theta^w_{ij^{'}}+w$. However, the happy woman have to pay some tax for accepting this bribe. Let $q\in[0,1]$ be  the fraction of the bribe she can get (after paying her tax). Her supposed cut for marrying $i$ is just $\theta^w_{ij^{'}}+qw$. Woman $j^{'}$ will believe and accept offer from man $i$ if two conditions are satisfied: the offer should be both
\begin{enumerate}
	\item {\em Competitive},   namely
	$\theta^w_{ij^{'}}+qw\geq \theta^w_{j^{'}j^{'}}$.
	\item {\em Trusted}, if woman $j^{'}$  believes that  man $i$ is motivated. This implies $\theta^m_{ij^{'}}-w\geq \theta^m_{ii^{'}}$.
\end{enumerate}
The two conditions above can be satisfied, and  the offer is {\em acceptable},  only if
\be\label{profitman}  q(\theta^m_{ij^{'}}-\theta^m_{ii^{'}})+\theta^w_{ij^{'}}-\theta^w_{jj^{'}}>0  \  . \ee
Symmetrically, man $i$ will accept an offer from a woman $j^{'}\not=i^{'}$ only if
\be\label{profitwoman}  q(\theta^w_{ij^{'}}-\theta^w_{ii^{'}})+\theta^m_{ij^{'}}-\theta^m_{jj^{'}}> 0  \  . \ee
The {\em utility} of the exchange $ii^{'}$ to $ij^{'}$ is, then defined by the minimum  $\Delta^{(q)}(i, j^{'})$ of (\ref{profitman}, \ref{profitwoman}) via  (\ref{introq}).

%Compare (\ref{intro0}) to (\ref{introq}).
\par

To understand the role of $p$, consider 
the chain of pairs exchanges
\begin{quote}
	$(i_1i^{'}_1\rightarrow i_1i^{'}_2),  \ldots  (i_{k-1}i^{'}_{k-1}\rightarrow i_{k-1} i^{'}_k), (i_ki^{'}_k)\rightarrow (i_ki^{'}_1)$ \ . 
\end{quote}

Each of the  pair exchange  $(i_l,i_l)\rightarrow (i_l,i_{l+1})$ yields a utility \\  $\Delta^{(q)}(i_l, i^{'}_{l+1})$ for  the new pair.
The lucky new pairs in this chain of couples exchange are those  who  makes a positive reward.
%$$ \left\{ (i_l, i_{l+1}) ; \ \ \ \ \Delta^{(q)}(i_l, i_{l+1})>0 \right\}\ .$$
The unfortunate new pairs are those who suffer a loss (negative reward). 
% $$ \left\{ (i_m, \tau(i_{m+1})) ; \ \ \ \ \Delta^{(q)}(i_m, \tau(i_{m+1})<0 \right\}\ , $$
The lucky pairs, whose interest is to activate this chain, are ready to compensate the unfortunate ones by contributing some of their gained utility.
%However, to minimize their contribution, they declare their profit to be the minimal one, so each such pair contributes
%$\Delta^{(q)}_\tau(i_l, \tau(i_{l+1}))$.
%The unfortunate pairs, on the other hand, reports their {\em maximal loss}   $-\Delta^{(q)}_\tau(i_m, \tau(i_{m+1}))$ (which is, by definition, positive).
The chain will be activated (and the original marriages will break down) if the mutual contribution of the fortunate pairs is enough to cover
{\it at least} the $p-$ part of the mutually   loss of utility
of the unfortunate pairs. This is the condition
$$ \sum_{\Delta^{(q)}(i_l, i^{'}_{l+1})>0}  \Delta^{(q)}(i_l, i^{'}_{l+1}) + p\sum_{\Delta^{(q)}(i_l, i^{'}_{l+1})<0} \Delta^{(q)}(i_l, i^{'}_{l+1})
\equiv  \sum_{l=1}^k\left[ \Delta^{(q)}(i_l, i^{'}_{l+1})\right]_p > 0 \   \ . $$
\vskip .3in
\begin{center}\begin{Ovalbox} {\it
			Stability by Definition \ref{defshapleyq}   grantees that no such chain is activated.
}\end{Ovalbox}\end{center}

	\section{Existence of stable marriage plans}
	In the general case of Assumption \ref{genass}, the existence of a stable matching follows from the following Theorem:
	\begin{theorem}\label{cornoemp}
		Let $W({\cal F})\subset \R^{2N}$ defined as follows:
		\\
		$(u_1, \ldots u_N, v_1, \ldots v_N)\in W({\cal F})$ , 
		
		$ \Leftrightarrow \exists$  an injection $ \tau:\I_m\rightarrow\I_w$  such that $(u_i, v_{i^{'}})\in {\cal F}(ii^{'})$ where $i^{'}=\tau(i)$ ,  $\ \forall \ i\in \I_m  \ . $
		Then there exists $(u_1, \ldots u_N, v_1, \ldots v_N)\in W({\cal F})$  such that
		\be\label{mercore} (u_i,v_{j^{'}})\in \R^2- {\cal F}_0(ij^{'}) \  \ee
		for any $(i,j^{'})\in\I_m\times \I_w$.
	\end{theorem}
	The set of vectors in $W({\cal F})$ satisfying (\ref{mercore}) is called {\em the core}. Note that the core is identified with the set of $\R^{2N}$ vector in $V({\cal F})$ which satisfy the condition $(u_i, v_{i^{'}})\in{\cal F}(ii^{'})$. Hence Definition \ref{deffea} can be recognized as the non-emptiness of the core, which is equivalent to the existence of a stable matching.
	\par
	Theorem \ref{cornoemp} is, in fact, a special case the celebrated Theorem of Scarf   \cite{Sc}  for  cooperative games, tailored to the marriage scenario  (see also \cite{Gi, bon}). As we saw, it can be applied to the fully non-transferable case (\ref{Fblockp}), as well as to the fully transferable case (\ref{stfultrans}).
	
	Theorem \ref{cornoemp} implies, in particular, the existence of stable marriage in the FNT case corresponding to $p=q=0$ or (\ref{Fblockp}), as well as for the FT case corresponding to $p=q=1$ or 
	(\ref{stfultrans}). 
	
%	In the re introduce a direct proof for these two cases. 
%	The proof is algorithmic in the FNT case and variational in the FT case. 
	\subsection{Gale-Shapley algorithm in the non-transferable case} 
	
	Here we describe the celebrated, constructive algorithm due to Gale and Shapley \cite{GSH}, is describe below:
	\begin{enumerate}
		\item At the first stage, each man $i\in\I_m$ proposes to the woman $j\in\I_w$ at the top of his list.
		At the end of this stage, some women got proposals (possibly more than one), other women may not get any proposal.
		\item At the second stage, each woman who got more than one proposal, bind the man whose proposal is most preferable according to her list (who is now engaged). She releases all the other men who proposed. At the end of this stage, the men's set $\I_m$ is composed of two parts: engaged and released.
		\item At the next stage each {\em released} man makes a proposal to the {\em next} woman in his preference list (whenever she is engaged or not).
		\item Back to stage 2.
	\end{enumerate}
	It is easy to verify that this process must end at a finite number of steps. At the end of this process all women and men are engaged. This is a stable matching!
	\par
	Of course, we could reverse the role of men and women in this algorithm. In both cases we get a stable matching. The algorithm we indicated is the one which is best from the men's point of view. In the case where the women propose, the result is best for the women. In fact
	\begin{theorem}\cite{MO}
		For any NT stable matching $\{ii^{'}\}$, the rank of the woman $i^{'}$ according to man $i$ is {\em at most} the rank of the woman matched to $i$ by the
		above, men proposing algorithm.
	\end{theorem}
	\subsection{Variational formulation in the fully transferable case}
	There are several equivalent definitions of stable marriage plan in the FT case. Here we introduces two of these. 
	
	Recall that if ${\cal F}$ is given by (\ref{calFFT}) the feasibility set $V({\cal F})$ (Definition \ref{deffea}) takes the form
	\be\label{Feass}V({\cal F}):=\{(u_1, \ldots v_N)\in\R^{2N}; \ \ \ u_i+v_{j^{'}}\geq \theta_{ij^{'}} \ \ \forall ij^{'}\in\I_m\times\I_w \}. \ee
	Recall also Definition \ref{chain} for cyclical monotonicity. 
	\begin{theorem}\label{matchingtrans}  $\{ii^{'}\}$ is a stable marriage  plan in the FT case  if and only if one of the following equivalent conditions  is satisfied:
		\begin{itemize}
			%	\item{i)} Optimality:  $(u^0_1, \ldots v^0_N)\in\R^{2N}$ is a minimizer of $\sum_1^Nu_i+v_i$ on $V({\cal F})$ and $u^0_i+v^0_{i^{'}}=\theta_{ii^{'}}$ for any $i=1\ldots N$. 
			\item Efficiency (or {\em maximal public utility}):   $\sum_{i=1}^N \theta_{ii^{'}}\geq \sum_{i=1}^N \theta_{i\sigma(i)}$ for any  marriage plans $\sigma:\I_m\rightarrow \I_w$.
			%\item{ii)} The sum of men and women cuts $\sum_{i=1}^N u_i + \sum_{j=1}^N v_j$ is minimal under the constraints $u_i+v_j\geq \theta_{ij}$ (i.e. $(u_i, v_j)\not\in F_0(i,j)$).
			\item $\{ii^{'}\}$ is cyclically monotone. 
			\item Optimality: The minimal sum $\sum_1^N u^0_i+v^0_i$ of  cuts in the feasibility set (\ref{Feass})  satisfies $u^0_i+v^0_{i^{'}}=\theta_{ii^{'}}$ (i.e $\{ u_1^0, \ldots v_N^0\}$ is in the core). 
		\end{itemize}
	\end{theorem} 
The efficiency characterization of stable marriage connects this notion with optimal transport and the celebrated {\em Monge Kantorovich} theory \cite{Mo, Ka, Vil}.  See also \cite{Gal}.

	Since the set of all bijections is finite and the maximum on a finite set is always achieved we obtain from the efficiency characterization. 
		\begin{cor}
		There always exists a stable marriage plan in the FT case. 
	\end{cor}
\begin{remark}
	As far as we know, the fully transferable case (\ref{Feass}) is the only case whose stable marriages are obtained by a variational argument. 
	\end{remark}

	\begin{proof} (of theorems \ref{matchingtrans})
		In Proposition \ref{t01} we obtained that FT stability implies cyclical monotonicity. We now prove that cyclical monotonicity implies efficiency. The proof follows the idea published originally by Afriat in 1963 \cite{effr}, and was introduced recently in a  much simpler form by Bresis \cite{Br}. 
		
		Let
		\be\label{pidef} -u^0_i:=\inf_{k-chains, k\in\N}\left( \sum_{l=1}^{k-1} \theta_{i_li^{'}_l}-\theta_{i_li^{'}_{l+1}}\right)
		+ \theta_{i_ki^{'}_k}-\theta_{i_ki^{'}}\ . \ee
		%%%%%%%%
		Let $\alpha>-u^0_i$ and consider a $k-$chain realizing
		\be\label{w1}\alpha>\left( \sum_{l=1}^{k-1} \theta_{i_li^{'}_l}-\theta_{i_li^{'}_{l+1}}\right)
		+ \theta_{i_ki^{'}_k}-\theta_{i_ki^{'}}\ee
		
		By cyclic monotonicity,  $ \sum_{l=1}^{k} \theta_{i_li^{'}_l}-\theta_{i_li^{'}_{l+1}}\geq 0$. 
		Since $i^{'}_{k+1}=i^{'}_1$,
		$$  \sum_{l=1}^{k-1} \theta_{i_li^{'}_l}-\theta_{i_li^{'}_{l+1}}\geq \theta_{i_ki^{'}_1}-\theta_{i_ki^{'}_k}$$
		so (\ref{w1}) implies
		$$ \alpha> \theta_{i_k, i^{'}_1} - \theta_{i_k, i^{'}}\geq 0 \ , $$
		in particular $u^0_i<\infty$.
		
		Hence, for any $j\in\I_m$
		\begin{multline} \alpha+\theta_{ii^{'}}-\theta_{ij^{'}} > \left( \sum_{l=1}^{k-1} \theta_{i_l, i^{'}_l}-\theta_{i_l, i^{'}_{l+1}}\right) \\
		+ \theta_{i_ki^{'}_k}-\theta_{i_ki^{'}}+\theta_{ii^{'}}-\theta_{ij^{'}} \geq -u^0_j\  \end{multline}
		where the last inequality follows by the substitution of  the $k+1-$cycle $i_1=i, i_2, \ldots, i_k, i_{k+1}=i$ in  (\ref{pidef}).
		%%%%%%%
		Since $\alpha$ is any number bigger than $-u^0_i$ it follows
		\be\label{ui-cii}-u^0_i+\theta_{ii^{'}}-\theta_{ij^{'}} \geq -u^0_j\ ,\ee
		for any pair $i,j\in\I_m$. Now, let $\sigma$ be any permutation in $\I_m$ and let $j=\sigma(i)$. Then
		\be\label{sumtausigma}-u^0_i+\theta_{ii^{'}}-\theta_{i\sigma(i^{'})} \geq -u^0_{\sigma(i)}\  . \ee
		Since  $\sigma$ is a bijection on $\I_m$ as well, so $\Sigma_{i=1}^N u^0_i=\sum_{i=1}^N u^0_{\sigma(i)}$. Then, sum (\ref{sumtausigma}) over $1\leq i\leq N$ to obtain
		$$\sum_{i=1}^N \theta_{ii^{'}} \geq \sum_{i=1}^N \theta_{i\sigma(i^{'})} \ , $$
		so $\{ii^{'}\}$ is an efficient marriage plan.
		\par
		To prove that any efficient  solution is stable, we define
		$v^0_j:= \theta_{jj^{'}}-u^0_j$ so
		\be\label{ui+vi} u^0_j+v^0_{j^{'}}=\theta_{jj^{'}} \ . \ee
		Then (\ref{ui-cii}) implies
		\be\label{ui+vj} u^0_i+v^0_{j^{'}}= u^0_i+\theta_{jj^{'}}-u^0_j\geq u^0_i-u^0_i+\theta_{ij^{'}}=\theta_{ij^{'}} \  \ee
		for any $i,j$. Thus, (\ref{ui+vi},\ref{ui+vj})  establish that $\{ii^{'}\}$ is a stable marriage via Definition \ref{deffea}.
		
		Finally, the optimality condition follows immedietly from the definition of the feasibility set
		$$ \sum_1^N u_i+v_{i^{'}}=\sum_1^N u_i+v_{\sigma(i)} \geq \sum_1^N \theta_{i\sigma(i)}$$
		for any bijection $\sigma:\I_m\rightarrow \I_w$ and from (\ref{ui+vi}). 
	\end{proof}
\subsection{On existence and non-existence of stable  fake promises}
%We now point out the following observation
%, based on Lemma  \ref{lemonotone} and Definition \ref{defshapleyq}:
\begin{theorem}\label{comparepq}
	If the matching $\{ii^{'}\}$ is $(p,q)-$stable, then it is also $(p^{'}, q^{'})-$stable for $p^{'}\geq p$ and $q^{'}\leq q$.\par
\end{theorem}

The proof of this Theorem follows from the definitions (\ref{introq}, \ref{pospdef}) and the following
\begin{lemma}\label{lemonotone}
	For any , $i\not=j$ and  $1\geq q>q^{'}\geq 0$,
	$$(1+q)^{-1}\Delta^{(q)}(i,j)>  (1+q^{'})^{-1}\Delta^{(q^{'})}(i,j) .  $$
\end{lemma}
\begin{proof}
	For $a,b\in\R$ and $r\in[0,1]$ define
	$$\Delta_r(a,b):= \frac{1}{2}(a+b)-\frac{r}{2}|a-b| \ . $$
	Observe that $\Delta_1(a,b)\equiv \min(a,b)$. In addition, $r \mapsto \Delta_r(a,b)$ is monotone not increasing in $r$.
	A straightforward calculation yields
	$$ \min(qa+b, qb+a)
	= \Delta_1(qa+b, qb+a)
	=(q+1)\Delta_{\frac{1-q}{1+q}}(a,b) \ , $$
	and the Lemma follows from the above observation, upon inserting $a=\theta_m(i,j)-\theta_m(i,i)$ and $b=\theta_w(i,j)-\theta_w(j, j)$.
\end{proof}

What can be said about the existence of s $(p,q)-$ stable matching in the general case? Unfortunately, we can prove now only a negative result:
\begin{prop}\label{propimpos}
	For any $1\geq q>p\geq 0$, a stable marriage does not exist unconditionally.
\end{prop}\begin{figure} 
\centering
\includegraphics[height=4.cm, width=10.cm]{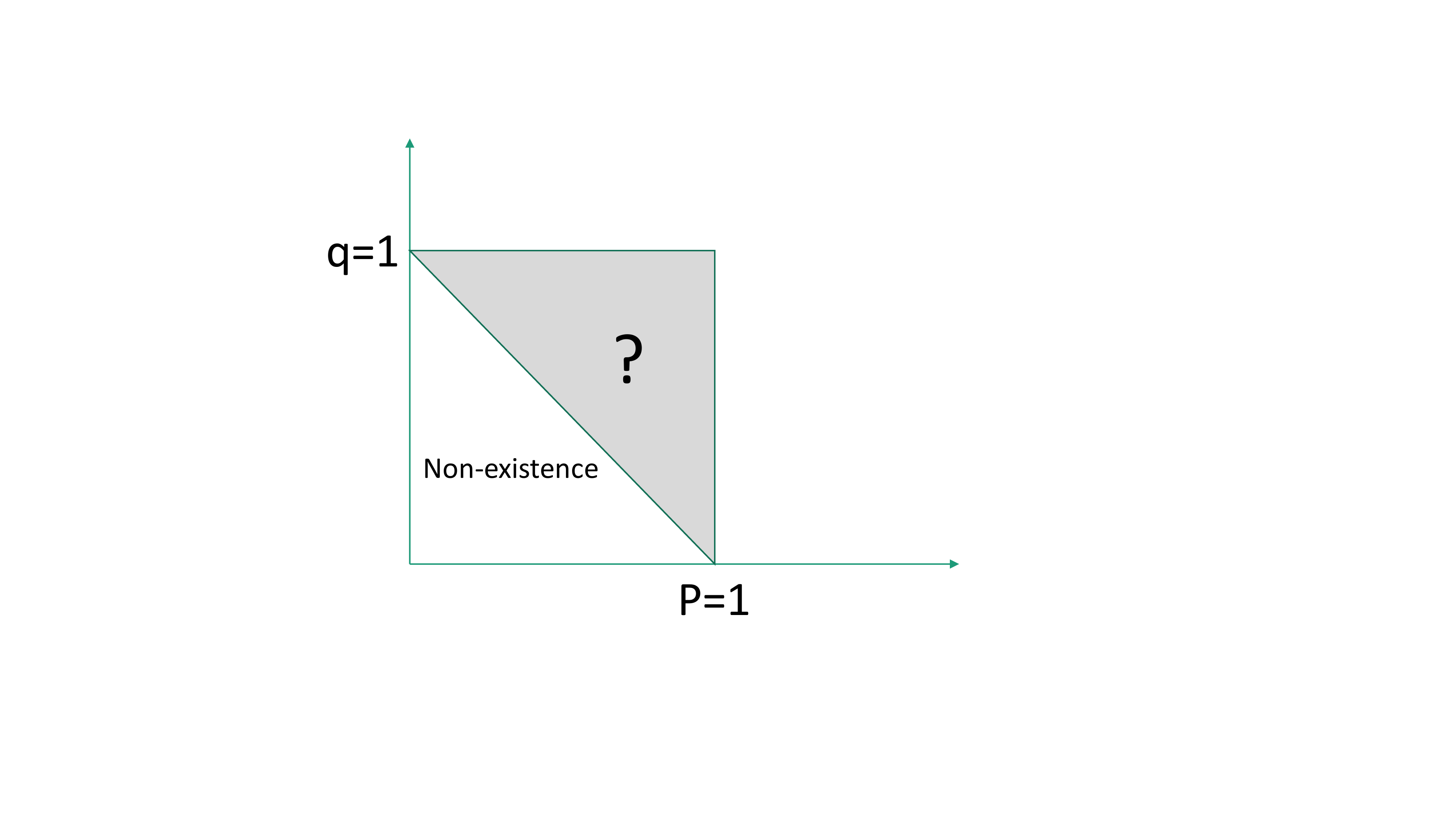}\\
\caption{Conjecture: Is there an unconditional existence of stable marriages in the gray area? }
\end{figure}
\begin{proof}
	We only need to present a counter-example. So, let $N=2$. To show that the matching $11^{'}, 22^{'}$ is not stable we have to show
	\be\label{t1}\left[ \Delta^{(q)}(1,2^{'})\right]_p+ \left[ \Delta^{(q)}(2,1^{'})\right]_p>0\ee
	while, to show that $12^{'}, 21^{'}$ is not stable we have to show
	\be\label{t2} \left[ \Delta^{(q)}(1,1^{'})\right]_p+ \left[ \Delta^{(q)}(2,2^{'})\right]_p>0 \ . \ee
	By definition (\ref{introq}) and Lemma \ref{lemonotone}
	$$ \Delta^{(q)}(1,2^{'})= (q+1)\Delta_r\left( \theta^m_{12^{'}}-\theta^m_{11^{'}}, \theta^w_{12^{'}}-\theta^w_{22^{'}}\right)$$
	$$ \Delta^{(q)}(2,1^{'})= (q+1)\Delta_r\left( \theta^m_{21^{'}}-\theta^m_{22^{'}}, \theta^w_{21^{'}}-\theta^w_{11^{'}}\right)$$
	where $r=\frac{1-q}{1+q}$. To obtain $\Delta^{(q)}(1,1^{'}), \Delta^{(q)}(2,2^{'})$ we just have to exchange man $1$ with man $2$, so
	$$ \Delta^{(q)}(2,2^{'})= (q+1)\Delta_r\left( \theta^m_{22^{'}}-\theta^m_{21^{'}}, \theta^w_{22^{'}}-\theta^w_{12^{'}}\right)$$
	$$ \Delta^{(q)}(1,1^{'})= (q+1)\Delta_r\left( \theta^m_{11^{'}}-\theta^m_{12^{'}}, \theta^w_{11^{'}}-\theta^w_{21^{'}}\right) \ . $$
	All in all, we only have 4 parameters to play with:
	$$ a_1:= \theta^m_{12^{'}}-\theta^m_{11^{'}}, \ \ a_2=\theta^w_{12^{'}}-\theta^w_{22^{'}} \ , $$
	$$ b_1=\theta^m_{21^{'}}-\theta^m_{22^{'}}, \ b_2=\theta^w_{21^{'}}-\theta^w_{11^{'}} \ ,  $$
	so the two conditions to be verified are
	$$ [\Delta_r(a_1, a_2)]_p+ [\Delta_r(b_1, b_2)]_p>0 \ \ ; \ \ [\Delta_r(-a_1, -b_2)]_p+ [\Delta_r( -b_1, -a_2)]_p >0 \ . $$
	Let us insert $a_1=a_2:=a>0$. $b_1=b_2:=-b$ where $b>0$. So
	$$[\Delta_r(a_1,a_1)]_p=a, \ \ \ [\Delta_r(b_1, b_2)]_p=-pb \ , $$
	while $\Delta_r(-a_1, -b_2)=\Delta_r(-b_1, -a_2)=\frac{b-a}{2}-\frac{r}{2}(a+b)$. In particular, the condition
	$\frac{a}{b}< \frac{1-r}{1+r}$ implies $[\Delta_r(-a_1, -b_2)]_p=[\Delta_r( -b_1, -a_2)]_p>0$ which verifies (\ref{t2}). On the other hand, if $a-pb>0$ then (\ref{t1}) is verified. Both conditions can be verified if $\frac{1-r}{1+r}>p$. Recalling $q=\frac{1-r}{1+r}$ we obtain the result.
\end{proof}
\begin{conj} If If $0< p<q< 1$ then there always exists a $(p,q)$ stable marriage (c.f. Fig [2]). 
	\end{conj}

%\begin{cor}\label{exctm}A stable matching according to Definition \ref{transmer} always exists.
%\end{cor}

%We shall return to Theorem \ref{matchingtrans} in section \ref{**}, after discussing a general existence Theorem:
%\section{Stable marriage by cheating}
%Here we define another notion of stable marriage, motivated by   a generalization of the cyclic monotonicity condition  of Theorem \ref{matchingtrans}.

\end{document}